\documentclass{article} 
\usepackage{latexsym}
\usepackage{amssymb,xspace,latexsym,epic,eepic,graphicx}
\usepackage{epsfig,amsmath}
\usepackage{colortbl}
\usepackage{fullpage}
\usepackage{stmaryrd}

\usepackage{latexsym}
\usepackage{xspace,latexsym,epic,eepic,graphicx}
\usepackage{epsfig}
\usepackage{amsmath,amssymb,bm}
\usepackage{subfigure, gastex}

\renewcommand{\l}{\ell}
\renewcommand{\phi}{\varphi}

\renewcommand{\epsilon}{\varepsilon}
\newcommand{\boxtheorem}{\hfill $\Box$}

\newcommand{\ignore}[1]{}
\newcommand{\OMIT}[1]{}

\newcommand{\A}{{\cal A}}

\newcommand{\V}{\mathbf{V}}

\renewcommand{\exp}{R}

\newcommand{\nexpr}{R}

\newcommand{\nex}{{\it n}}

\newcommand{\sem}[1]{\llbracket #1 \rrbracket}

\newcommand{\ptime}{{\sc Ptime}}

\newcommand{\pspace}{{\sc Pspace}}

\newcommand{\vv}[1]{{\small \texttt{#1}}}

\newtheorem{theorem}{Theorem}[section]
\newtheorem{lemma}[theorem]{Lemma}

\newtheorem{proposition}[theorem]{Proposition}
\newtheorem{example}[theorem]{Example}

\newcommand{\com}{\%}
\newcommand{\fin}{\&}
\newcommand{\finbra}{\$}
\newcommand{\trans}{\text{\it trans}}
\newcommand{\posi}{\text{\it pos}}

\renewcommand{\A}{A}

\newenvironment{proof}
{\par\noindent\textit{Proof: }}
{\boxtheorem}

\title{Containment of Nested Regular Expressions} 
\author{Juan Reutter}
\date{}

\begin{document}

\maketitle 

\begin{abstract}
Nested regular expressions (NREs) have been proposed as a powerful formalism for querying RDFS graphs, but research in a more general graph database context has been scarce, and static analysis results are currently lacking. In this paper we investigate the problem of containment of NREs, and show that it can be solved in PSPACE, i.e., the same complexity as the problem of containment of regular expressions or regular path queries (RPQs).\end{abstract}



\section{Introduction}


Graph-structured data has become pervasive
in 
data-centric applications.
Social networks, bioinformatics, astronomic databases, 
digital libraries, Semantic Web, and
linked government data, are only a few examples  
of applications in which 
structuring data as graphs is, simply, essential.

Traditional relational query languages do not
appropriately cope with the querying problematics raised by 
graph-structured data. The reason for this is twofold. First, 
in the context of graph databases one is typically interested
in {\em navigational} queries, i.e. queries that traverse the
edges of the graph while checking for the existence of paths satisfying
certain conditions. However, most relational query
languages, such as SQL, are not designed to deal with this kind of
recursive queries \cite{AHV95}. Second, current graph database
applications tend to be massive in size (think, for instance, of social
networks or astronomic databases, that may store terabytes
of information). Thus, one can immediately dismiss any 
query language that cannot be
evaluated in polynomial time (or even in linear time!). 
But then even the core of the usual relational query
languages -- {\em conjunctive} queries (CQs) -- does not satisfy
this property. In fact, parameterized complexity analysis tells us that
-- under widely-held complexity theoretical analysis -- 
CQs over graph databases cannot be evaluated in time $0(|G|^c \cdot
f(|\phi|))$, where $c \geq 1$ is a constant and $f : \mathbb{N} \to
\mathbb{N}$ is a computable function 
\cite{PY99}.

This raises a need for languages that are specific for the graph database context. 
The most commonly used {\em core} of these languages are the so-called {\em regular path queries}, or RPQs \cite{CMW87}, 
that specify the existence of paths between nodes, with the
restriction that the labels of such path belong to a regular language.
The language of RPQs was later extended with the ability to traverse edge backwards, providing them with a 
{\em 2-way} functionality. This gives rise to the notion of 2RPQs \cite{CGLV00}. 

Nested regular expressions are a graph database language that aims to extend 
the possibility of using regular expressions, or 2-way regular expressions, for querying graphs 
with an {\em existential test}
operator $[(\cdot)]$, also known as {\em nesting} operator, similar to the one in XPath \cite{GKP05}. 
This class of expressions was proposed in~\cite{PAG10} 
for querying Semantic Web data, and have received a fair deal of attention in 
the last years \cite{RDFS,BPR12,BPR13}. 

We say that 
Here we study the problem of containment of NREs, 
which is the following problem: 

\begin{center}
\fbox{
\begin{tabular}{ll}
{Problem}: & {\sc NREContainment}\\ 
{Input}: & 
NREs $Q_1$ and $Q_2$ over $\Sigma$.  \\
{Question}: & Is $Q_1 \subseteq Q_2$?
\end{tabular}
}
\end{center}

Note that we study this problem for the restricted case when all the possible input graphs 
are {\em semipaths}. The general case will be shown in an extended version of the manuscript. 

\section{Preliminaries}
\label{sec:prelim} 

\subsection{Graph Database and queries}
\noindent
{\bf Graph databases}. 
Let $\V$ be a countably infinite set of \emph{node ids}, and $\Sigma$ 
a finite alphabet. A \emph{graph database} $G$
over $\Sigma$ is a pair $(V,E)$, where $V$ is a
finite set of node ids (that is $V$ is a finite subset of $\V$) and $E
\subseteq V \times \Sigma \times V$. 
That is, $G$ is an {\em edge-labeled} directed graph, where the fact that
$(u,a,v)$ belongs to $E$ means that there is an edge from node
$u$ into node $v$ labeled $a$. 
For a graph database $G=(V,E)$,
we write $(u,a,v)\in G$ whenever $(u,a,v)\in E$.

\begin{figure*}[t!]
\begin{center}
\hspace*{-60pt} \input{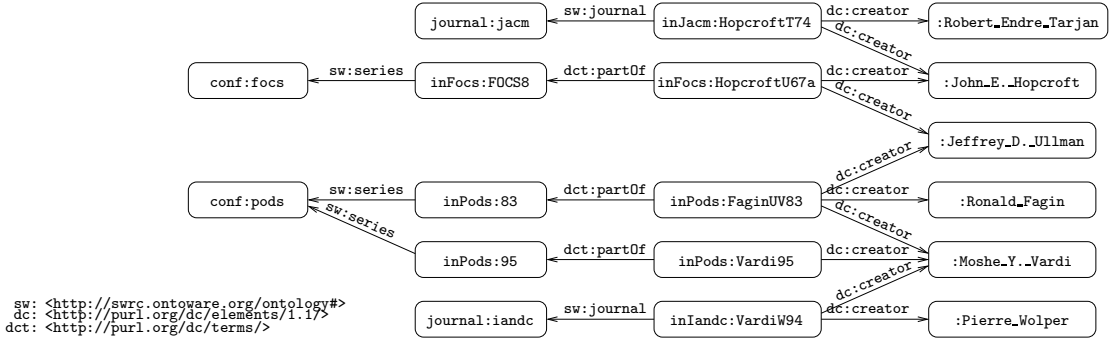}
\end{center}
\vspace*{-15pt}
\caption{\normalfont A fragment of the RDF Linked Data representation of DBLP~\cite{dblp.l3s} 
available at \vv{http://dblp.l3s.de/d2r/}}
\label{fig:main}
\end{figure*}

\medskip 
\noindent
{\bf Nested Regular Expressions}. 

The language of {\em nested regular
expressions} (NREs) were first proposed in \cite{PAG10} for querying
Semantic Web data. Next we formalize the language of nested regular
expressions in the context of graph databases.

Let $\Sigma$ be a finite alphabet. 
The NREs over 
$\Sigma$ extend classical regular expressions with an
\emph{existential nesting test} operator $[\, \cdot \,]$ (or just nesting operator, for short), 
and an \emph{inverse} operator $a^-$, over each $a\in\Sigma$.
The syntax of NREs is given by the following grammar:
\begin{multline*}
\nexpr \; := \;\;
\varepsilon \;\mid\; a\; (a\in \Sigma) \;\mid\; a^-\; (a\in \Sigma) \;\mid\; \nexpr\cdot \nexpr \\ 
\nexpr^* \; \mid \; \nexpr + \nexpr \;\mid\; [\nexpr]  \label{eq:grammar}
\end{multline*}
As it is customary, we use $\nex^+$ as shortcut for $\nex\cdot \nex^*$.

Intuitively, 
NREs specify pairs of node ids in a graph database, 
subject to the existence of a 
path satisfying a certain regular condition among them.
That is, each NRE $\nexpr$ defines a binary relation $\sem{\nexpr}_G$ when evaluated over a
graph database $G$.  This binary relation is defined inductively as follows,
where we assume that $a$ is a symbol in $\Sigma$, and $n$, $n_1$ and
$n_2$ are arbitrary NREs:  
\begin{eqnarray*}
\sem{\varepsilon}_G & = & \{(u,u)\mid u\text{ is a node id in }G\} \\
\sem{a}_G & = & \{(u,v)\mid (u,a,v)\in G\} \\ \sem{a^-}_G & = &
\{(u,v)\mid (v,a,u)\in G\} \\ \sem{\nex_1\cdot\nex_2}_G & = &
\sem{\nex_1}_G\circ \sem{\nex_2}_G \\ 
\sem{\nex_1 + \nex_2}_G & = & \sem{\nex_1}_G \cup \sem{\nex_2}_G \\
\sem{\nex^*}_G & = & \sem{\varepsilon}_G \cup \sem{\nex}_G \cup
\sem{\nex\cdot \nex}_G \cup \sem{\nex\cdot \nex\cdot \nex}_G \cup
\cdots \\ \sem{\, [\nex]\, }_G & = & \{(u,u)\mid \text{there exists
}v\text{ s.t.\ }(u,v)\in \sem{\nex}_G\}. 
\end{eqnarray*}
Here, the symbol $\circ$ denotes the usual composition of binary
relations, that is, $\sem{\nex_1}_G\circ \sem{\nex_2}_G=\{(u,v)\mid$
there exists $w$ s.t.\ $(u,w)\in \sem{\nex_1}_G$ and $(w,v)\in
\sem{\nex_2}_G\}$.


\begin{example}\label{ex:nest}
Let $G_1$ be the graph database in Figure~\ref{fig:main}.
The following is a simple NRE that matches all pairs $(x,y)$
such that $x$ is an author that published a paper in conference $y$:
\[
n_1\; = \; \vv{creator}^- \cdot \vv{partOf} \cdot \vv{series}
\]
For example the pairs $(\vv{:Jeffrey\_D.\_Ullman},\vv{conf:focs})$ and
$(\vv{:Ronald\_Fagin},\vv{conf:pods})$ are in $\sem{n_1}_G$.
Consider now the following expression that 
matches pairs $(x,y)$ such that $x$ and $y$
are connected by a \emph{coautorship sequence}:
\[
n_2\;=\; (\vv{creator}^- \cdot \vv{creator})^+
\]
For example the pair $(\vv{:John\_E.\_Hopkroft},\vv{:Pierre\_Wolper})$, is
in $\sem{n_2}_G$. 
Finally the following expression matches all pairs $(x,y)$ such that $x$ and $y$ are
connected by a coautorship sequence that only considers conference papers:
\[
n_3\;=\; (\vv{creator}^- \cdot [\, \vv{partOf} \cdot \vv{series} ] \cdot \vv{creator})^+
\]
Let us give the intuition of the evaluation of this expression. Assume that
we start at node $u$. The (inverse) edge $\vv{creator}^-$ makes us to navigate from $u$ to
a paper $v$ created by $u$. Then the existential test $[\, \vv{partOf} \cdot \vv{series} ]$
is used to check that from $v$ we can navigate to a conference (and thus, $v$ is a conference paper).
Finally, we follow edge $\vv{creator}$ from $v$ to an author $w$ of $v$.
The $(\cdot)^+$ over the expression allows us to repeat this sequence several times.
For instance, $(\vv{:John\_E.\_Hopkroft},\vv{:Moshe\_Y.\_Vardi})$ is in $\sem{n_3}_G$, but
$(\vv{:John\_E.\_Hopkroft},\vv{:Pierre\_Wolper})$ is not in $\sem{n_3}_G$.
\end{example}

\smallskip \noindent {\bf {\em Complexity and expressiveness of NREs}} \,
The following result, proved in~\cite{PAG10}, shows a remarkable
property of NREs. It states that the query
evaluation problem for NREs is not only polynomial in \emph{combined}
complexity (i.e. when both the database and the query are given as
input), but also that it can be solved linearly in
both the size of the database and the expression. 
Given a graph database $G$ and an NRE
$\nexpr$, we use $|G|$ to denote the size of $G$ (in terms of the
number of egdes $(u,a,v)\in G$), and $|\nexpr|$ to denote the size of
$\nexpr$.  

\begin{sloppypar}
\begin{proposition}[from
\cite{PAG10}]\label{prop:poly} Checking, given a graph database $G$, a
pair of nodes $(u,v)$, and an NRE $\nexpr$, whether $(u,v)\in
\sem{\nexpr}_G$, can be done in time $O(|G|\cdot|\nexpr|)$.
\end{proposition}
\end{sloppypar} 

On the expressiveness side, NREs
subsume several important query languages for graph databases. For
instance, by disallowing the inverse operator $a^{-}$ and the nesting
operator $[\,\cdot\,]$ we obtain the class of {\em regular path
queries} (RPQs) \cite{CMW87,MW95}, while by only disallowing the
nesting operator $[\,\cdot\,]$ we obtain the class of RPQs with {\em
inverse} or 2RPQs \cite{CGLV00}. (In particular, both expressions
$n_1$ and $n_2$ in Example \ref{ex:nest} are 2RPQs). In turn, 
NREs allow for an important increase in expressive power over
those languages. For example, it can be shown that NRE expression
$n_3$ in Example \ref{ex:nest} cannot be expressed without 
the nesting operator $[\,\cdot\,]$, and hence it is not expressible in
the language of 2RPQs (c.f. \cite{PAG10}).

On the other hand, the class of NREs fails capturing more expressive
languages for graph-structured data that combine navigational
properties with quantification over node ids. Some of the most
paradigmatic examples of such languages are the classes of {\em
conjunctive} RPQs and 2RPQs, that close RPQs and 2RPQs, respectively,
under conjunctions and existential quantification. Both classes of
queries have been studied in depth, as they allow
identifying complex patterns over graph-structured data
\cite{CM90,FLS98,CGLV00a}.   



\section{Containment of NREs over paths}
\label{sec-words}
\medskip
\noindent 
{\bf Problem Definition}.
It is convenient for the proof to explain first how NRE's are used to represent regular 
languages over $\Sigma$, and how to represent these languages 
using alternating two way automata. 
Let us begin with some notation. 

Along the proof we assume that $\Sigma$ includes all {\em reverse symbols}. 
More precisely, if $\Sigma'$ is an alphabet, we work instead with the alphabet 
$\Sigma = \Sigma' \cup \{a^- \mid a \in \Sigma'\}$. 
Let $G = (V,E)$ be a graph over $\Sigma$. A semipath in $G$ is a sequence 
$u_1,a_1,u_2,a_2,\dots, u_m,a_m,u_{m+1}$,, where each $u_i$ belongs to $V$, each 
$a_i$ belongs to $\Sigma$, and 
for each $u_i,a_i,u_{i+1}$, we have that $(u_i,a_i,u_{i+1})$ belongs to $E$, 
if $a_i$ is not a reverse symbol, and $(u_{i+1},a_i,u_{i})$ belongs to $E$ is 
$a_i$ is a reverse symbol, i.e., of form $a^-$ for some $a \in \Sigma'$. 
A semipath is simple if all of its nodes are distinct. 
Finally, a graph $G$ {\em resembles a (simple) semipath} if there is a (simple) semipath $\pi$ in $G$ of the form above such that 
the nodes of $G$ are precisely $\{u_1,\dots,u_n\}$ and the edges of edges of $G$ are precisely 
those that witness the above definition. 

As we have mentioned, we study {\sc NREContainment} only when the input graphs are semipaths. 
We are now ready to describe our goal 
which is to show that the following problem is in \pspace: 
Given NREs $Q_1$ and $Q_2$ over $\Sigma$, decide wether $\sem{Q_1}_G \subseteq \sem{Q_2}_G$, for all 
graphs $G$ over $\Sigma$ such that $G$ resembles a simple semipath. 
In what follows, we refer to this problem as 
{\sc SP-NREContainment}.

\subsection{Alternating 2-way finite automata}

Following \cite{LLS84}, an {\em Alternating 2-way finite automaton}, or A2FA for short, is a tuple $\A = (Q,q_0,U,F,\Sigma,\delta)$, where 
$Q$ is the set of states, $U \subseteq Q$ is a set of {\em universal} states, $q_0$ is the initial state, $F \subseteq Q$ is the 
set of final states, $\Sigma$ is the input alphabet (we also use symbols $\com$ and $\fin$ not in $\Sigma$ 
as the start and end markers of the string), and the transition function is $\delta: Q \times (\Sigma \cup \{\com, \fin\}) \rightarrow 2^{Q \times \{-1,0,1\}}$. 

The numbers $-1,0,1$ in the transition stand for moving back, staying and moving forward, respectively. 
The input is delimited with $\com$ at the beginning and $\fin$ at the end. For convenience, we assume that the automaton starts 
in state $q_0$ while reading the symbol $\fin$ of the string. 

\medskip
\noindent
{\bf Semantics} 
Semantics are given in terms of computation trees over instantaneous descriptions. An {\em instantaneous description} (ID) 
is a triple of form $(q,w,i)$, where $q$ is a state, $w$ is a word in $\com \sigma^* (\epsilon \|\ \fin)$ and $1 \leq i \leq |w|+1)$. 
Intuitively, it represent the state of the current computation, the string it has already read, and the current position of the automata. 
An ID is {\em universal} if $q \in U$ and {\em existential} otherwise, and accepting IDs are of form $(q,w,|w+1|)$ for 
$w \in \com \Sigma^*\fin$ and $q \in F$. 

Let $w = a_1,\dots,a_n$, for each $a_i \in \Sigma \cup \{\com,\fin\}$. The transition relation $\Rightarrow$ is defined as follows: 
\begin{itemize}
\item $(q,w,i) \Rightarrow (p,w,i)$, if $(p,0) \in \delta(q,a_i)$ and $1 \leq i \leq n$; 
\item $(q,w,i) \Rightarrow (p,w,i+1)$, if $(p,1) \in \delta(q,a_i)$ and $1 \leq i \leq n$; and 
\item $(q,w,i) \Rightarrow (p,w,i-1)$, if $(p,-1) \in \delta(q,a_i)$ and $1 < i \leq n$.  
\end{itemize}

A {\em computation tree} $\Pi$ of an A2FA $\A = (Q,q_0,U,F,\Sigma,\delta)$ is a finite, nonempty tree with each of its nodes $\pi$ labelled 
with an ID $\l(\pi)$, and such that 
\begin{enumerate}
\item If $\pi$ is a non-leaf node and $\l(\pi)$ is universal, let $I_1,\dots,I_k$ be all IDs such that $\l(\pi) \Rightarrow I_j$ for each 
$1 \leq j \leq k$. Then $\pi$ has exactly $k$ children $\pi_1,\dots,\pi_k$, where $\l(\pi_j) = I_j$; and 
\item If $\pi$ is a non leaf node and $\l(\pi)$ is existential, then $\pi$ has exactly one child $\pi'$ such that $\l(\pi) \Rightarrow \l(\pi')$.  
\end{enumerate}

Finally, an {\em accepting computation tree} of $\A$ over $w$ is a computation tree $\Pi$ whose root is labelled with $(q_0,\com w \fin, |w| + 2)$ and 
each of it leaves are labelled with an accepting ID. 

We need the following theorem. It follows immediately from the results in \cite{LLS84}:

\begin{proposition}
\label{prop-nonempt-a2fa}
Given a A2FA $\A$, it is \pspace-complete to decide wether the language of $\A$ is empty. 
\end{proposition}

\subsection{Proof of {\sc SP-NREContainment}}

The idea is to code acceptance of strings by NREs using alternating 2-way automata. More precisely, given an 
NRE $\exp$, we construct an A2FA $\A_\exp$ such that the language of $\A_\exp$ corresponds, in a precise sense, 
to all those words $w$ such that $\sem{\exp}_G$ is nonempty for all those graphs $G$ that resemble the simple semipath 
$w$. 

\medskip
\noindent
{\bf Construction of $\A_\exp$}. 
We define the translation by induction, all states are existential unless otherwise noted.  
Along the construction, we shall be {\em marking}, in each step, a particular state of the automata. 
We use this mark in the construction.  Furthermore, for the sake of readability we shall include 
$\epsilon$-transitions. This is without loss of generality, as they can be easily simulated with two transitions (and an extra state), 
the first moving forward, and the second backwards. 

\begin{itemize}
\item If $\exp = a$, then $A_\exp = (\{q_0,q_f,q_r\}, \emptyset, q_0, q_f, \Sigma,\delta)$, with 
$\delta$ defined as: 
\[
\begin{array}{lcl}
\delta(q_0,a) & = & \{(q_f,1), (q_r,-1) \} \\
\delta(q_0,b) & = & \{(q_r,-1)\}\, , \text{for each $b \in \Sigma$, $b \neq a$}  \\
\delta(q_r,a^-) & = & \{(q_f,0), \} \\ 
\end{array}
\]

State $q_r$ and the two way functionality is added so that the automaton correctly accepts 
when the input is a word of form $\Sigma^* a^- \Sigma^*$ (See \cite{CGLV00} for a thorough explanation of this machinery). 
Moreover, state $q_f$ is {\em marked}. 

\item Similarly, if $\exp = a^-$, then $A_\exp = (\{q_0,q_f,q_r\}, \emptyset, q_0, q_f, \Sigma,\delta)$, with 
$\delta$ defined as: 
\[
\begin{array}{lcl}
\delta(q_0,a^-) & = & \{(q_f,1), (q_r,-1) \} \\
\delta(q_0,b) & = & \{(q_r,-1)\}\, , \text{for each $b \in \Sigma$, $b \neq a^-$}  \\
\delta(q_r,a) & = & \{(q_f,0), \} \\ 
\end{array}
\]

State $q_f$ is {\em marked}.

\item Case when $\exp = \exp_1 + \exp_2$. Let $A_{\exp_i} = (Q^i, U^i, q_0^i, F^i, \Sigma, \delta^i)$, for $i = 1,2$, 
and assume that $q_m^i$ is the marked state from $A_{\exp_i}$. 
Define $A_\exp = (Q, U, q_0, F, \Sigma, \delta)$, where $Q = \{q_0,q_f\} \cup Q^1 \cup Q^2$, $U = U^1 \cup U^2$, 
$F = \{q_f\} \cup (F^1 \setminus \{q_m^1\}) \cup (F^2 \setminus \{q_m^2\})$ and $\delta = \delta^1 \cup \delta^2$, plus transitions 
\[
\begin{array}{lcl}
\delta(q_0, \epsilon) & = & \{(q_0^1,0),(q_0^2,0)\} \\
\delta(q_m^1,\epsilon) & = & \{(q_f,0)\} \\
\delta(q_m^2,\epsilon) & = & \{(q_f,0)\} \\
\end{array}
\]

For each $i = 1,2$, remove al marks from $A_{\exp_i}$, and 
{\em mark} state $q_f$.

\item In the case that $\exp = \exp_1 \cdot \exp_2$, let $A_{\exp_i} = (Q^i, U^i, q_0^i, F^i, \Sigma, \delta^i)$, for $i = 1,2$, 
and assume that $q_m^i$ is the marked state from $A_{\exp_i}$. 
For each $i = 1,2$, remove all marks from $A_{\exp_i}$. 
Define $A_\exp = (Q, U, q_0, F, \Sigma,  \delta)$, where $Q = \{q_0,q_f\} \cup Q^1 \cup Q^2$, $U = U^1 \cup U^2$, 
$F = \{q_f\} \cup (F^1 \setminus \{q_m^1\}) \cup (F^2 \setminus \{q_m^2\})$ and $\delta = \delta^1 \cup \delta^2$, plus transitions 

\[
\begin{array}{lcl}
\delta(q_0, \epsilon) & = & \{(q_0^1,0)\} \\
\delta(q_m^1,\epsilon) & = & \{(q_0^2,0)\} \\
\delta(q_m^2,\epsilon) & = & \{(q_f,0)\} \\
\end{array}
\]

For each $i = 1,2$, remove al marks from $A_{\exp_i}$, and 
{\em mark} state $q_f$.

\item For $\exp = \exp_1^*$, let $A_{\exp_1} = (Q^1,U^1,q_0^1, F^1,\Sigma, \delta^1)$, 
and assume that $q_m^1$ is the marked state from $A_{\exp_1}$. 

Define $A_\exp = (Q, U^1, q_0, F, \Sigma, \delta)$, where $Q = \{q_0,q_f\} \cup Q^1$, 
$F = \{q_f\} \cup (F^1 \setminus \{q_m^1\})$ and $\delta = \delta^1$ plus transitions 
\[
\begin{array}{lcl}
\delta(q_0, \epsilon) & = & \{(q_0^1,0)\} \\
\delta(q_0^1, \epsilon) & = & \{(q_f,0)\} \\
\delta(q_m^1,\epsilon) & = & \{(q_f,0), (q_0^1,0)\} \\
\end{array}
\]

Remove al marks from $A_{\exp_1}$, and 
{\em mark} state $q_f$.

\item When $\exp = [\exp_1]$, let $A_{\exp_1} = (Q^1,U^1, q_0^1, F^1, \Sigma, \delta^1)$, 
and assume that $q_m^1$ is the marked state from $A_{\exp_1}$. 
Then 
$A_\exp = ( Q, U^1, q_0, F, \Sigma, \delta)$, where $Q = \{q_0,p,q_2,q_f\} \cup Q^1$, $U = U^1 \cup \{p\}$,
$F = \{q_f\} \cup F^1$ and $\delta = \delta^1$, plus transitions 
\[
\begin{array}{lcl}
\delta(q_0, \epsilon) & = & \{(p,0)\} \\
\delta(p,\epsilon) & = & \{(q_f,0), (q_i^1,0)\} \text{ (recall that } p \text{ is a universal state)} \\
\delta(q_m^1,a) & = & \{(q_m^1,1)\} \text{ for each }a \in \Sigma \\
\end{array}
\]

Remove al marks from $A_{\exp_1}$, and 
{\em mark} state $q_f$.
\end{itemize}

Let $\A_\exp = (Q,q_0,U,F,\Sigma,\delta)$ be as constructed by this algorithm. To finish our construction 
we need to allow $\A_\exp$ to (non deterministically) move backwards from the end of the word, until it reaches a suitable 
starting point for the computation, and allow every final state to reach the end of the word in its computation. 
Formally, we define $\A_\exp' = (Q \cup \{q_0'\},q_0',U,F,\Sigma \cup \{\fin\},\delta')$, where 
$\delta'$ contains all transitions in $\delta$ plus transitions 
$\delta(q_0',a) = \{(q_0,0),(q_0',-1)\}$ for each $a \in \Sigma \cup \{\fin\}$ and $\delta(q_f,a) = (q_f,1)$ for 
each $a \in \Sigma$ and $q_f \in F$.

Notice that the above construction can be computed in polynomial 
time with respect to $\exp$. Furthermore, 
let $q_m$ be the marked (final) state of $A_\exp'$. From its construction, 
it is clear that every accepting computation tree $\Pi$ of $A_\exp$ on input 
$w$ will have the following form: 
(1) For some $1 \leq i \leq |w|$ there is a single path from the root to a node $\pi_s$ such that $\l(\pi) = (q_0,w,i)$ and no ancestor 
of $\pi$ is labelled with an ID using a state different from $q_0'$; and (2) there is 
some $1 \leq j \leq |w|$ such that $\pi_f, \pi_f',\pi_f'',\dots $ is the maximal path of nodes labelled with 
$(q_m,w,j)$, $(q_m,w,j+1),\dots,(q_m,w,|w|+1)$, i.e., the father of $(q_m,w,j)$ is not labelled with an ID using state 
$q_m$. 
Property (1) represents the automaton searching for its starting point, and (2) represents the 
end of the computation of the part of $\A_\exp'$ that is representing the non-nesting part of $\exp$. 
We denote such nodes $\pi_s$ and $\pi_f$ as the {\em tacit start} and {\em tacit ending} of $\Pi$.  
With this definitions we can show the following.

\begin{lemma}
\label{lem-1-correctstring}
Let $S$ be a graph over $\Sigma$ that is a semipath, $w$ the label of the path $S$, 
and $\exp$ a NRE. Then a pair $(u_i,u_j)$ belongs to $\sem{\exp}_S$ if and only if 
there is an accepting computation tree of $A_\exp'$ on input $w$ whose tacit start is 
labelled with $(q_0,w,i)$ and whose tacit ending is labelled with  $(q_m,w,j)$. 
\end{lemma}

\begin{proof}
Let $S$ be the semipath $u_1,a_1,u_2,a_2,\dots, u_m,a_m,u_{m+1}$, and therefore 
$w = a_1 \cdot \dots \cdot a_m$, and let $A_\exp' = (\Sigma, Q, U, q'_0, \delta', F)$ constructed as explained above. 

\smallskip

For the {\bf only if direction}, assume that $\sem{\exp}_S$ contains the pair $(u_i,u_j)$, $1 \leq i,j \leq m+1$. 
We prove the above statement by induction on $\exp$. 
\begin{itemize}
\item If $\exp = a$, for some $a \in \Sigma$, and $(u_i,u_j) \in \sem{\exp}_S$, then either $j = i+1$ and the edge 
$(u_i,a,u_j)$ is in $S$, or $j = i-1$ and the edge $(u_j,a^-,u_i)$ is in $S$. In the former case the existence of a computation tree 
is obvious, for the latter case observe that one could use the transitions $(q_0,w,i) \Rightarrow (q_r,w,i-1)$, and then since 
$(u_j,a^-,u_i)$ is in $S$ we follow transition $(q_r,w,i-1) \Rightarrow (q_f,w,i-1)$. 

\item Case for $\exp = a^-$ is analogous to the previous one 

\item If $\exp = \exp_1 + \exp_2$ and $(u_i,u_j) \in \sem{\exp}_S$, then $(u_i,u_j) \in \sem{\exp_k}_S$ for $k = 1$ or $k = 2$, 
which entails a proper accepting computation tree for $A_{\exp_1}$ ($A_{\exp_2}$) on input $w$. The statement follows immediately 
from the construction of $A_\exp$. 

\item If $\exp = \exp_1 \cdot \exp_2$ and $(u_i,u_j) \in \sem{\exp}_S$, then there is a node $u_k$ of $S$ such that 
$(u_i,u_k) \in \sem{\exp_1}_S$ and $(u_k,u_j) \in \sem{\exp_2}_S$. Assume that the initial and marked nodes of 
$A_{\exp_1}$ and $A_{\exp_2}$ are $q_0^1$, $q_m^1$ and $q_0^2$, $q_m^2$, respectively. From the induction hypothesis we have that there 
are accepting computation trees for $A_{\exp_1}$ and $A_{\exp_2}$ whose tacit starts are $(q_0^1,w,i)$ and $(q_0^2,w,k)$, 
respectively, and the tacit ending of the first tree is labelled with $(q_m^1,w,k)$. Since $A_\exp'$ has, by construction, 
the pair $(q_0^2,0)$ in $\delta(q_m^2,\epsilon)$, we can cut the first tree in its tacit ending and 
plug in the computation tree for $A_{\exp_2}$, starting from its tacit start, which proves the statement. 

\item The case when $\exp = \exp_1^*$ goes along the same lines as the concatenation, except this time we may
have to plug in a greater number of computation trees. 

\item Finally, if $\exp = [\exp_1]$ and $(u_i,u_j) \in \sem{\exp}_S$, then $u_i = u_j$, and there is some $u_k$ such that 
$(u_i,u_k) \in \sem{\exp_1}_S$. Let $q_p$ be the universal state in $A_\exp'$ that is not in $A_{\exp_1}$. Then the only transitions associated to 
$q_p$ are $\delta(q_p,\epsilon) = \{(q_0^1,0), (q_f,0)\}$, with $q_f$ being the only marked (final) state of $A_\exp$. 
Our accepting computation tree for $\A_\exp'$ has a path from the root to the tacit start, then a node labeled $(q_p,w,i)$ with children 
$(q_0^1,w,i)$ and $(q_f,w,i)$, with the computation tree for $A_{\exp_1}$ (starting from its tacit start) plugged into the first of these children. 
\end{itemize}

\medskip 

For the {\bf if direction}, assume that there is an accepting computation tree of 
$A_\exp$ on input $w$ whose tacit start is labelled with 
$(q_0,w,i)$ and with its tacit ending labelled with $(q_m,w,j)$. We now prove that $(u_i,u_j)$ belong 
to $\sem{\exp}_S$. The proof is again by induction 
\begin{itemize}
\item For the base case when $\exp = a$ (proof for $\exp = a^-$ is analogous), there are two options for an accepting computation of $A_\exp$. 
Either it is of form $(q_0,w,i) \Rightarrow (q_f,w,i+1)$, in which case $j = i+1$ and 
$a_i = a$, or it is of form $(q_0,w,i) \Rightarrow (q_r,w,i-1) \Rightarrow (q_f,w,i-1)$, 
in which case $j = i-1$ and $a_i = a^-$. For both cases we obtain that $(u_i,u_j) \in \sem{\exp}_S$. 

\item When $\exp = \exp_1 + \exp_2$, by the construction of $A_\exp$, any computation tree of $A_\exp$ can be prunned from its tacit start to 
obtain a computation tree for one of $A_{\exp_1}$ or $A_{\exp_2}$, from where the statement easily follows. 

\item When $\exp = \exp_1 \cdot \exp_2$, we can similarly obtain computation trees for $A_{\exp_1}$ and $\A_{\exp_2}$, and then 
conclude that $(u_i,u_j)$ belong 
to $\sem{\exp}_S$. Same hold when $\exp = \exp_1^*$, except in this case we obtain multiple computation trees for $\exp_1$. 

\item Finally, if $\exp = [\exp_1]$ and there is an accepting computation tree of 
$A_\exp$ on input $w$ whose tacit start is labelled with 
$(q_0,w,i)$ and with its tacit ending labelled with $(q_m,w,j)$, from the construction of $A_\exp$ the top part of the computation tree is of form 
$(q_0,w,i) \Rightarrow (q_p,w,i) \Rightarrow (q_0^1,w,i) , (q_m,w,i)$, where $q_p$ is the only universal state of $A_\exp$ not in $A_{\exp_1}$, 
and $q_0^1$ is the initial state of $A_{\exp_1}$. Then the part of the computation tree that follows from node  $(q_0^1,w,i)$ comprises 
a computation tree 
for $A_{\exp_1}$, i.e., there is a $u_k$ such that $(u_i,u_k) \in \sem{\exp_1}_S$. This entails that $(u_i,u_i) \in \sem{\exp}_S$.
\end{itemize}
\end{proof}

\medskip
\noindent
{\bf Proof for containment}
For our algorithm of containment, we need to be a little more careful, since for a word $w$ accepted by $A_\exp$ 
it is not necessarily the case that $u_i$ and $u_j$ are the start and finish nodes of the semipath $S$. Thus, we have to 
distinguish the start/end of the word with the actual piece that is framed by nodes $u_i$ and $u_j$ in the semipath. 
In order to do that, we augment $\Sigma$ with two extra symbols $S,E$. Furthermore, if 
$A_\exp = (\Sigma, Q, U, q_0', \delta, F)$, and $q_m$ is the marked state of $A_\exp$, 
we construct $A_\exp^{S,E} = (\Sigma \cup \{S,E\}, Q \cup \{q_0^S, q_f^E\}, U, q_0^S, \delta^{S,E}, (F \setminus \{q_m\})\cup \{q_f^E\})$, 
where $\delta^\$$ is defined as follows: 
for each state $q \in Q \setminus U$, we add the pair $(q,1)$ to $\delta(q,S)$ and $\delta(q,E)$, if $q$ is not $q_0$ or $q_m$, the pair 
$(q_f^E,1)$ to $\delta(q_m,E)$, $(q_0,1)$ to $\delta(q_0^S,S)$, plus the pair $(q_0^S,-1)$ to each $\delta(q_0^S,a)$ for $a \in \Sigma \cup \{E\}$ 
and $(q_f^E,1)$ to each $\delta(q_f^E,a)$ for $a \in \Sigma$. 

The intuition is the following. Let $\exp$ be an NRE and $A_\exp$ be the A2FA constructed as above. Now assume that there is a semipath 
$w = u_1,a_1,u_2,\dots,u_n,a_n,u_{n+1}$ and nodes $u_i,u_j$ such that $(u_i,u_j) \in \sem{\exp}_S$. By the above Lemma, we have that 
there is a computation tree for $A_\exp$ that tacitly starts in $(q_0,w,i)$ and tacitly ends in $(q_m,w,j)$. The idea of the symbols $S$ and 
$E$ is to specifically mark the tacit start and end of the piece $a_i,\dots,a_{j-1}$ labeling the semipath between $u_i$ and $u_j$. Thus, in 
this case, $\A_\exp^{S,E}$ accepts the word $a_1 \cdots a_{i-1} S a_i \cdots a_{j-1} E a_j \cdots a_n$. It uses intuitively the same computation 
tree mentioned before, except now it moves backwards in state $q_0^S$ until symbol $S$ is reached, then proceeds with the computation, and 
the marked branch now ends in $q_m^E$ instead of $Q_m$, after checking there is a symbol $E$ after $a_{j-1}$. With this intuition, it is 
straightforward to show: 

\begin{lemma}
\label{lem-2-correctstring}
Let $w = u_1,a_1,u_2,\dots,u_n,a_n,u_{n+1}$ be a graph over $\Sigma$ that is a simple semipath, $w = a_1,\dots,a_n$ the label of the path $w$, 
and $\exp$ a NRE. Then a pair $(u_i,u_j)$ belongs to $\sem{\exp}_S$ if and only if 
$\A_\exp^{S,E}$ accepts the word $a_1 \cdots a_{i-1} S a_i \cdots a_{j-1} E a_j \cdots a_n$.  
\end{lemma}


We can now state our algorithm for solving {\sc SP-QueryContainment}. 
On input NREs $\exp_1$ and $\exp_2$, we perform the following operations: 
\begin{enumerate}
\item Compute an NFA $\A^{S,E}$ that accepts only those words over 
$(\Sigma \cup \{S,E\})^*$ of form $w_1Sw_2Ew_3$, for each $w_1,w_2,w_3$ in $\Sigma^*$.  
\item Compute $\A_{\exp_1}^{S,E}$ and $A_{\exp_2}^{S,E}$ as explained above. 
\item Compute the A2FA $A^c = (A_{\exp_2}^{S,E})^c$ whose language is the complement of $A_{\exp_2}^{S,E}$ 
\item Compute the A2FA $A$ whose language is the intersection of the languages $A^{S,E}$, $A_{\exp_1}^{S,E}$ and 
$A^c$. 
\item Check that the language of $A$ is empty
\end{enumerate}

We have seen how to perform the second step in polynomial time, and steps (1), (3), (4) can be easily 
performed in \ptime\ using standard techniques from automata theory. Finally, Proposition \ref{prop-nonempt-a2fa} shows that step (5) can be performed in 
\pspace. Thus, all that is left to prove is 
that the language of the resulting automata $A$ is empty if and only if $\exp_1 \subseteq \exp_2$. 

Assume first that $\exp_1 \subseteq \exp_2$, and 
assume for the sake of contradiction that there is a word $w \in L(A)$. 
We have that $w$ must be of form $a_1 \cdots a_{i-1} S a_i \cdots a_{j-1} E a_j \cdots a_n$, and 
$w$ is accepted by $A_{\exp_1}^{S,E}$, but not by $A_{\exp_2}^{S,E}$. Let $S$ be a graph consisting of the semipath 
$u_1, a_1, u_2,\dots,u_n,a_n,u_{n+1}$. By Lemma \ref{lem-2-correctstring}, nodes $(u_i,u_j) \in \sem{\exp_1}_S$, and 
thus by our assumption $(u_i,u_j)$ must belong to $\sem{\exp_2}_S$, but this would imply, again by the lemma, 
that $w$ is accepted by $A_{\exp_2}^{S,E}$. 

On the other hand if $L(A)$ is empty but $\exp_1 \not\subseteq \exp_2$, then for some graph 
$S = u_1,a_1,u_2,\dots,u_n,a_n,u_{n+1}$ that is a semipath and nodes $u_i,u_j$ it is the case that 
$(u_i,u_j) \in \sem{\exp_1}_S$, yet $(u_i,u_j) \notin \sem{\exp_1}_S$. By Lemma \ref{lem-2-correctstring}, 
we have that $w= a_1 \cdots a_{i-1} S a_i \cdots a_{j-1} E a_j \cdots a_n$ is accepted by 
$A_{\exp_1}^{S,E}$, and it is not accepted by $A_{\exp_2}^{S,E}$, thus belonging to $A^c$. Since clearly 
$w$ is also in the language of $A^{S,E}$, this means that $w$ belongs to $L(A)$, which is a contradiction. 


\section{Containment of NREs}

We now turn to the general problem. Let us begin with a few
technical definitions. 

\medskip
\noindent
{\bf $k$-branch semipaths}. 
Of course, when dealing with general graph databases, we cannot longer use the construction of section \ref{sec-words}, 
since it is specifically tailored for strings (or graph that look like paths). 
Nevertheless, we shall prove below that, even for the general case, 
we only need to focus on a very particular type of graphs, that we call here {\em $k$-branch semipaths}. 

Fix a natural number $k$. A $k$-branch domain $D$ is a prefix closed subset of 
$1 \cdot \{1,\dots,k\}^*$ such that 
\begin{enumerate}
\item no element in $D$ is of form 
$\{1,\dots,k\}^* \cdot i \cdot \{1,\dots,k\}^* \cdot j \cdot \{1,\dots,k\}^*$ with $i > j$. 
\item If $w \cdot i$ belongs to $D$ and there is a different element with prefix $w \cdot i$ in $D$, 
then $w \cdot i \cdot i$ belongs to $D$.  
\end{enumerate}

A {\em $k$-branch semipath} over $\Sigma$ is a tuple $T = (D,E)$, where $D$ is a $k$-branch domain, 
and $E \subseteq D \times \Sigma \times D$ respects the structure of the tree: for each $u$ in $D$ there is a single 
edge to each of its {\em children} $u \cdot i, u \cdot j, u \cdot \ell, \dots$ that belong to $D$, and there are no outgoing edges from the 
leaves of $D$.  Note that $k$-branch domains have essentially $k$ types of elements: Each element of the class $[j]$ 
comprises all string of form $s \cdot j$, for $s \in D$, and these elements  
can have children only of classes $[j],\dots,[k]$. We call each of this classes a {\em branch} of the semipath. Note that we impose that 
the element $1$ must be always be the root of a $k$-branch semipath (instead of the usual $\epsilon$. 

\medskip
\noindent
{\bf Canonical graphs for NREs}. 
Let $\exp$ be a NRE. We define the {\em nesting depth} of $\exp$ according to the following inductive definition: 
\begin{itemize}
\item The nesting depth of $a$ or $a^-$ is $1$, for $a \in \Sigma$. 
\item If $\exp_1$ has nesting depth $i$, then $\exp_1^*$ has nesting depth $i$.
\item If $\exp_1$ has nesting depth $i$ and $\exp_2$ has nesting depth $j$, then 
$\exp_1 \cdot \exp2$ and $\exp_1 + \exp_2$ have nesting depth $\max(i,j)$ 
\item If $\exp_1$ has nesting depth $i$ then $[\exp_1]$ has nesting depth $i+1$. 
\end{itemize}

Let $\exp$ be an NRE and $T = (D,E)$ a $k$-branch semipath. 
We now need to define when $T$ is {\em canonical} for $\exp$. We do it in an inductive fashion. 

\begin{itemize}
\item If $\exp = a$, for $a \in \Sigma$, then $T$ is canonical if it contains only two elements, $u$ and $u \cdot i$, and 
the edge from $u$ to $u \cdot i$ is labelled with $a$. 
\item If $\exp = \exp_1 + \exp_2$, then $T$ is canonical for $\exp$ is it is canonical for $\exp_1$ or for $\exp_2$. 
\item If $\exp = \exp_1 \cdot \exp_2$, then $T$ is canonical for $\exp$ if there exists an element $w$ in $T$ such that, 
if we define $T_1$ as the $k$-branch semipath induced by the set of elements $\{w\} \cup \{u \mid w$ is not a prefix of $u\}$ and 
$T_2$ the $k$-branch semipath induced by the set $\{u \mid w$ is a prefix of $u\}$, then 
$T_1$ is canonical for $\exp_1$ and $T_2$ is canonical for $\exp_2$  
(in other words, $T$ is the concatenation of $T_1$ and $T_2$). 
\item If $\exp = \exp_1^*$, then $T$ is canonical for $\exp$ if it contains a single node, or otherwise for some $n \geq 1$ 
there are elements $w_1, w_1 \cdot w_2, \ldots, w_1 \cdot w_2 \cdots w_n$ that define series of 
$T_1,\dots,T_n$ induced subgraphs (as in the previous case), and 
each of these are canonical for $\exp_1$. 
\item If $\exp = [\exp]$, then $T$ is canonical for $[\exp]$ if it is canonical for $\exp$. 
\end{itemize} 

The following proposition highlights the importance of canonical graphs in our context. 
\begin{proposition}
Let $\exp_1$ and $\exp_2$ be NREs, and assume that the nesting depth of $\exp_1$ is $k$. Then 
$\exp_1$ is not contained in $\exp_2$ if and only if there is a $k$-branch semipath $T$ that is canonical 
for $\exp_1$ and two nodes of $T$ such that $(n_1,n_2) \in \sem{\exp_1}_T$ yet $(n_1,n_2) \notin \sem{\exp_2}_T$. 
\end{proposition}

\begin{proof}[Sketch]
$(\Leftarrow)$: By definition.

$(\Rightarrow)$: Follows by {\em monotonicity} of NREs, using techniques similar to those in \cite{BPR13}. The idea is as follows. 
Assume that $\exp_1$ is not contained in $\exp_2$. Then there is a graph $G$ and two nodes of $G$ 
such that $(n_1,n_2) \in \sem{\exp_1}_G$ yet $(n_1,n_2) \notin \sem{\exp_2}_G$. By carefully following the construction of $\exp_1$, one 
can prune $G$ into a $k$-branch semipath $T$ (recall that $k$ is the nesting depth of $\exp$) that is canonical for $\exp_1$, and such that 
it still holds that $(n_1,n_2) \in \sem{\exp_1}_T$. 
Since NREs are monotone and $T \subseteq G$ it must be the case that $(n_1,n_2) \notin \sem{\exp_2}_T$
\end{proof}

\subsection{Main Proof}

We now proceed with the \pspace\ upper bound for {\sc NREContainment}. Let $\exp_1$ and $\exp_2$ be NREs over $\Sigma$ that 
are the inputs to this problem, and consider a symbol $\finbra$ not in $\Sigma$. The roadmap of the proof is the following. 

\begin{enumerate}
\item We will first show an encoding scheme $\trans$ that transforms every $k$-branch semipath into a string 
over alphabet $\Gamma_k = \{1,\dots,k\} \times (\Sigma \cup \{\finbra\}) \times \{1,\dots,k\}$.
\item Afterwards, we show that one can construct, given an NRE $\exp$, an automaton $A_\exp$ over $\Gamma_k$ that accepts, 
in a precise sense, all encodings of $k$-branch semipaths that satisfy $\exp$. 
\item Finally we proceed just as in Section \ref{sec-words}, deciding whether $\exp_1 \subseteq \exp_2$ by taking 
the complement of $A_{\exp_2}$, intersecting it with $A_{\exp_1}$, and checking that the resulting automaton defines the empty language. 
\end{enumerate}

\medskip
\noindent
{\bf Coding $k$-branch semipaths as strings}

In the following we show how $k$-branch semipaths over an alphabet $\Sigma$ can be coded into strings. 
Let $\finbra$ be a symbol not in $\Sigma$, and $K = \{1,\dots,k\}$. We use the alphabet 
$\Gamma_k = K \times (\Sigma \cup \{\finbra\}) \times K$, and define the translation inductively. 
Note that we maintain the assumption that strings begin and end with symbols $\com$ and $\fin$, respectively. 
When $K$ is understood from context, we simply talk about $\Gamma$. 

Let $T = (D,E)$ be a $k$-branch semipath. We define $\trans(T)$ as $\com \cdot \trans(1) \cdot \fin$, where $1$ is the root element 
of $T$. For each element in $D$, the relation $\trans$ is defined as follows: 

\begin{itemize}
\item If $w \cdot i$ is a leaf in $T$, then $\trans(w \cdot i)$ is the single symbol $(i,\finbra,i)$. 
\item Otherwise, assume that the children of $w \cdot i$ are $w \cdot i \cdot \ell_1 \dots, w \cdot i \cdot \ell_p$, and the 
label of each edge from $w \cdot i$ to $w \cdot i \cdot \ell_j$ is $a_j$. 
Then $$\trans(w \cdot i) = (i,a_2,\ell_2) \cdot \trans(w \cdot i \cdot \ell_2) \cdot \cdots \cdot (i,a_p,\ell_) \cdot 
\trans(w \cdot i \cdot \ell_p) \cdot (i,a_1,\ell_1) \cdot \trans(w \cdot i \cdot \ell_1)$$
\end{itemize}


Note that the number of characters in $\trans(T)$ is precisely the sum of the number of edges and the number of 
leaves of the $k$-branch semipath $T$. 
We need a way to relate positions in trees with position in their translations. Formally, this is done via a 
function $\posi$, that assigns to every  node $w$ in a $k$-branch semipath $T$, the position in 
$\trans(T)$ that corresponds to the point where the substring $\trans(w)$ starts in $\trans(T)$. Note then that all positions 
in $\trans(T)$ will have a pre-image in $T$, except for those positions that are immediately after a symbol of form 
$(i,\finbra,i)$ in $\trans(T)$. 

Finally, the following proposition shows that the languages of strings represented by semipaths is regular. Moreover, an alternating automaton 
representing this language can be constructed in polynomial time with respect to $k$. 

\begin{proposition}
For each $k \geq 1$ there is an alternating automaton that accepts the language of all strings over $\Gamma_k$ that 
are encodings of a $k$-branch semipath. 
\end{proposition}

\begin{proof}
It is usefull to construct first an automaton that accepts the complement of the language in the statement of the proposition. 
Since AFA can be complemented in polynomial time, the proof then follows. 

Fix then a number $k \geq 1$. We now sketch construct an AFA $\A_k$ that accepts all strings over $\Gamma$ which are not encodings of 
a $k$-banch semipath. Essentially, we need $\A_k$ to check for the following: 

\begin{enumerate}
\item The string uses symbol $(i, \finbra, j)$ for some $i \neq j$ in $\{1,\dots,k\}$, or any symbol $(i,a,j)$ for some $a \in \Sigma$ and 
$i > j$. 

\item The string does not start with $\com \cdot (1,a,1)$ for some $a \in \Sigma$ 

\item The string does not ends with the symbol $(1, \finbra,1) \cdot \fin$. 

\item There is more than one appearance of the symbols $\com$ or $\fin$. 

\item For every $i, j \in \{1,\dots,k\}$ and every $a \in \Sigma$, a symbol 
$(i,a,j)$ appears without a forthcoming symbol $(j,\finbra,j)$. 

\item For every $i \in \{1,\dots,k\}$ and every $a \in \Sigma$, the symbol $(i,\finbra,i)$ appears without 
a preceding symbol $(j,a,i)$, for some $j \in \{1,\dots,k\}$

\item There are two appearances of symbols of form $(i,\finbra,i)$ without any symbol of form $(j,a,i)$ in between them, for some 
$j < i$ and $a \in \Sigma$.

\item For some $i,j,\ell \in \{1,\dots,k\}$, $j >i$, $\ell \leq i$ and $a,b \in \Sigma$ there are two symbols $(i,a,j)$ and $(i,b,j)$ 
between symbols $(\ell,a,i)$ and $(i,b,i)$, for some $a',b' \in \Sigma$

\item For every $i, j \in \{1,\dots,k\}$ and every $a \in \Sigma$,  
at some point after the symbol $(i,a,j)$ and before the symbol $(j,\finbra,j)$ there 
is a symbol of form $(i',a',j')$, for $a \in \Sigma$, with either $i'$ or $j'$ strictly lower than $j$. 

\item For every $i, j \in \{1,\dots,k\}$ and every $a \in \Sigma$, after a subword that starts with 
the symbol $(i,a,j)$, has only symbols of form $(\ell,b,\ell')$ for $\ell,\ell' \geq j$ and $b \in \Sigma$, and ends with 
symbol $(j,\finbra,j)$; there is a symbol of form $(p,c,p')$ with $p \neq i$ and $c \in \Sigma \cup \{\finbra\}$. 

\end{enumerate}

It is now straightforward to construct such automaton. Furthermore, since complementation in alternating automata 
can be performed in polynomial time, the proof follows. 
\end{proof}

\medskip
\noindent
{\bf Alternating automata for NREs}

All that remains for the \pspace-upper bound is to show how one can
construct, given an NRE $\exp$, an
A2FA $\A_\exp$ that accepts all strings over $\Gamma$ that
are encodings of $k$-branch semipaths that satisfy $\exp$.

\medskip
\noindent
{\bf Construction of $\A_\exp$: }

We start with some technical definitions. For every $i,j \in \{1,\dots,k\}$, define the 
language $L_{(i,j)}$ as follows: 

$$L_{(i,j)} = (i,a,j) \cdot \big(\{(\ell,b,\ell') \in \Gamma \mid \ell, \ell \geq j \text{ and } b \in \Sigma\}\big)^* \cdot (j,\finbra,j). $$

Intuitively, each $L_{(i,j)}$ defines path that departs from level $i$ to level $j$ in the $k$-branch semipath. 

The main technical difficulty in this construction is to allow the
automaton to navigate through the encoding of the $k$ branch semipath.
This is mostly captured by the base cases of our inductive
construction the idea is that one now has to allow the automata to skip words  
of form $L_{(i,j)}$ when choosing the next symbol (or when looking for it when reaching backwards), 
or in fact allow it to jump to a different branch in the semipath. We shall therefore make repeated use of 
the languages $L_{(i,j)}$. We also define, for each $1 \leq i \leq k$, the language 

$$B_i = \big(L_{(i,i)} + L_{(i,i+1)} + \cdots + L_{(i,k)}\big)^*.$$

Finally, we also work with a 2 way automaton $A_{B_j^-}$ that {\em read backwards} from symbol 
$(j,\finbra,j)$ to the first symbol of form $(i,a,j)$ for some $i < j$. More precisely, 
$B_j^- = (Q,q_0, \emptyset,\{q_f\},\Gamma,\delta)$, where $Q = \{q_0,q_1,q_2,q_f\}$ and 
$\delta$ is as follows: 
\begin{itemize}
\item For each symbol $a \in \Gamma$, $\delta(q_0,a) = \{q_1,-1\}$. This moves the automaton a step backwards, so we can 
start checking our language. 
\item In addition, $\delta(q_1,(j,\finbra,j)) = \{(q_2,-1\}$. This piece checks that we start with $(j,\finbra,j)$. 
\item For each $a \in \Sigma \cup \{\finbra\}$ and $k,k' \geq j$, $\delta(q_2,(k,a,k')) = \{(q_2,-1)\}$. This forces the automaton to 
loop in this state if one does not find the start of the branch. 
\item Finally, for each $a \in \Sigma$ and $i < j$, $\delta(q_2,(i,a,j)) = \{(q_f,0)\}$. This simply checks that the last symbol marks the beginning 
of the branch. 
\end{itemize}

With this definitions we can start describing the construction. Let $\exp$ be an NRE. The automaton $\A_\exp$ 
for $\exp$ is as follows. 
\begin{itemize}
\item If $\exp = a$ for some $a \in \Sigma'$. For each $1 \leq i \leq k$, let $A_{B_i}$ be a copy of an automaton accepting 
$B_i$, using fresh states, and assume that their initial and final states, respectively, are $p_0^i$ and $p_f^i$. Furthermore, 
for each $1 \leq j \leq k$ create a fresh copy of the automaton $A_{B_j^-}$, with initial and final states $(p_0^j)^-$ and $(p_f^j)^-$, respectively. 
then $\A_\exp = (Q,q_0,\emptyset,\{q_f\},\Gamma,\delta)$, where $Q$ contains $\{q_0,q_r^1,q_r^2,q_f\}$ plus all the states of the 
automata $A_{B_i}$ and $A_{B_j^-}$ for $i,j \in \{1,\dots,k\}$, and $\delta$ contains, apart from all the transitions in 
the $A_{B_i}'s$ and $A_{B_j^-}'s$, the following transitions: 
\begin{itemize}
\item For each $1 \leq i \leq j \leq k$, $\delta(q_0,(i,a,j)) = \{(q_f,1),(q_r^1,-1)\}$
|item For each $1 \leq i \leq j \leq k$ and $b \in \Sigma$ $(b \neq a)$, $\delta(q_0,(i,b,j)) = \{(q_r^1,-1)\}$
\item $\delta(q_0,\epsilon) = \{(p_0^1,0), \ldots, (p_0^k,0)\}$
\item $\delta(p_f^i,\epsilon) = \{(q_0,0)\}$ for each $1 \leq i \leq k$. 
\item For each $1 \leq i \leq j \leq k$, $\delta(q_r^1,(i,a^-,j)) = \{(q_r^2,0)\}$
\item $\delta(q_r^2, \epsilon) = \{(q_f, 0), (p_f^1)^-,0),\ldots,((p_f^k)^-,0)\}$
\item $\delta((p_f^i)^-,\epsilon) = \{(q_r^2,0)\}$ for each $1 \leq i \leq k$
\end{itemize}

State $q_f$ is {\em marked}.

\item If $\exp = a^-$ for some $a \in \Sigma'$. For each $1 \leq i \leq k$, let $A_{B_i}$ be a copy of an automaton accepting 
$B_i$, using fresh states, and assume that their initial and final state, respectively, are $p_0^i$ and $p_f^i$. Furthermore, 
for each $1 \leq j \leq k$ create a fresh copy of the automaton $A_{B_j^-}$, with initial and final states $(p_0^j)^-$ and $(p_f^j)^-$, respectively. 
then $\A_\exp = (Q,q_0,\emptyset,\{q_f\},\Gamma,\delta)$, where $Q$ contains $\{q_0,q_r^1,q_r^2,q_f\}$ plus all the states of the 
automata $A_{B_i}$ and $A_{B_j^-}$ for $i,j \in \{1,\dots,k\}$, and $\delta$ contains, apart from all the transitions in 
the $A_{B_i}'s$ and $A_{B_j^-}'s$, the following transitions: 
\begin{itemize}
\item For each $1 \leq i \leq j \leq k$, $\delta(q_0,(i,a^-,j)) = \{(q_f,1),(q_r^1,-1)\}$
|item For each $1 \leq i \leq j \leq k$ and $b \in \Sigma$ $(b \neq a^-)$, $\delta(q_0,(i,b,j)) = \{(q_r^1,-1)\}$
\item $\delta(q_0,\epsilon) = \{(p_0^1,0), \ldots, (p_0^k,0)\}$
\item $\delta(p_f^i,\epsilon) = \{(q_0,0)\}$ for each $1 \leq i \leq k$. 
\item For each $1 \leq i \leq j \leq k$, $\delta(q_r^1,(i,a,j)) = \{(q_r^2,0)\}$
\item $\delta(q_r^2, \epsilon) = \{(q_f, 0), (p_f^1)^-,0),\ldots,((p_f^k)^-,0)\}$
\item $\delta((p_f^i)^-,\epsilon) = \{(q_r^2,0)\}$ for each $1 \leq i \leq k$
\end{itemize}

State $q_f$ is {\em marked}.

\item Case when $\exp = \exp_1 + \exp_2$. Let $A_{\exp_i} = (Q^i, U^i, q_0^i, F", \Gamma, \delta^i)$, for $i = 1,2$, 
and assume that $q_m^i$ is the marked state from $A_{\exp_i}$. 
Define $A_\exp = ( Q, U, q_0, F, \Gamma, \delta)$, where $Q = \{q_0,q_f\} \cup Q^1 \cup Q^2$, $U = U^1 \cup U^2$, 
$F = \{q_f\} \cup (F^1 \setminus \{q_m^1\}) \cup (F^2 \setminus \{q_m^2\})$ and $\delta = \delta^1 \cup \delta^2$, plus transitions 
\[
\begin{array}{lcl}
\delta(q_0, \epsilon) & = & \{(q_0^1,0),(q_0^2,0)\} \\
\delta(q_m^1,\epsilon) & = & \{(q_f,0)\} \\
\delta(q_m^2,\epsilon) & = & \{(q_f,0)\} \\
\end{array}
\]

For each $i = 1,2$, remove al marks from $A_{\exp_i}$, and 
{\em mark} state $q_f$.

\item In the case that $\exp = \exp_1 \cdot \exp_2$, let $A_{\exp_i} = (Q^i, U^i, q_0^i, F", \Gamma, \delta^i)$, for $i = 1,2$, 
and assume that $q_m^i$ is the marked state from $A_{\exp_i}$. 
For each $i = 1,2$, remove all markings from $A_{\exp_i}$. 
Define $A_\exp = (Q, U, q_0, F, \Gamma, \delta)$, where $Q = \{q_0,q_f\} \cup Q^1 \cup Q^2$, $U = U^1 \cup U^2$, 
$F = \{q_f\} \cup (F^1 \setminus \{q_m^1\}) \cup (F^2 \setminus \{q_m^2\})$ and $\delta = \delta^1 \cup \delta^2$, plus transitions 

\[
\begin{array}{lcl}
\delta(q_0, \epsilon) & = & \{(q_0^1,0)\} \\
\delta(q_m^1,\epsilon) & = & \{(q_0^2,0)\} \\
\delta(q_m^2,\epsilon) & = & \{(q_f,0)\} \\
\end{array}
\]

For each $i = 1,2$, remove al marks from $A_{\exp_i}$, and 
{\em mark} state $q_f$.

\item For $\exp = \exp_1^*$, let $A_{\exp_1} = (Q^1,U^1,q_0^1, F^1, \Gamma, \delta^1, F^1)$, 
and assume that $q_m^1$ is the marked state from $A_{\exp_1}$. 

Define $A_\exp = (Q, U^1, q_0, F, \Gamma, \delta)$, where $Q = \{q_0,q_f\} \cup Q^1$, 
$F = \{q_f\} \cup (F^1 \setminus \{q_m^1\})$ and $\delta = \delta^1$ plus transitions 
\[
\begin{array}{lcl}
\delta(q_0, \epsilon) & = & \{(q_0^1,0)\} \\
\delta(q_0^1, \epsilon) & = & \{(q_f,0)\} \\
\delta(q_m^1,\epsilon) & = & \{(q_f,0), (q_0^1,0)\} \\
\end{array}
\]

Remove al marks from $A_{\exp_1}$, and 
{\em mark} state $q_f$.

\item When $\exp = [\exp_1]$, let $A_{\exp_1} = (Q^1,U^1, q_0^1, F^1, \Gamma, \delta^1)$, 
and assume that $q_m^1$ is the marked state from $A_{\exp_1}$. 
Then 
$A_\exp = (Q, U^1, q_0, F, \Gamma, \delta)$, where $Q = \{q_0,p,q_2,q_f\} \cup Q^1$, $U = U^1 \cup \{p\}$,
$F = \{q_f\} \cup F^1$ and $\delta = \delta^1$, plus transitions 
\[
\begin{array}{lcl}
\delta(q_0, \epsilon) & = & \{(p,0)\} \\
\delta(p,\epsilon) & = & \{(q_f,0), (q_i^1,0)\} \text{ (recall that } p \text{ is a universal state)} \\
\delta(q_m^1,a) & = & \{(q_m^1,1)\} \text{ for each }a \in \Gamma \\
\end{array}
\]
Remove al marks from $A_{\exp_1}$, and 
{\em mark} state $q_f$.

\end{itemize} 

Let $\A_\exp = (Q,q_0,U,F,\Gamma,\delta)$ be as constructed by this algorithm. To finish our construction 
we need to allow $\A_\exp$ to (non deterministically) move backwards from the end of the word, until it reaches a suitable 
starting point for the computation, and allow every final state to reach the end of the word in its computation. 
Formally, we define $\A_\exp' = (Q \cup \{q_0'\},q_0',U,F,\Gamma,\delta')$, where 
$\delta'$ contains all transitions in $\delta$ plus transitions 
$\delta(q_0',a) = \{(q_0,0),(q_0',-1)\}$ for each $a \in \Gamma$ and $\delta(q_f,a) = (q_f,1)$ for 
each $a \in \Gamma$. 
In the remainder of the proof, when speak of the automata for $\exp$ we refer to this last automaton $A_\exp'$, 
even if we use the clearer $A_\exp$ instead. 

The rest of the proof goes along the same lines as the version for semipaths. 
Notice that the above construction can be computed in polynomial 
time with respect to $\exp$. Furthermore, 
let $q_m$ be the marked state of $A_\exp'$. From its construction, 
it is clear that every accepting computation tree $\Pi$ of $A_\exp$ on input 
$w$ will have the following form: 
(1) For some $1 \leq i \leq |w|$ there is a single path from the root to a node $\pi_s$ such that $\l(\pi) = (q_0,w,i)$ and no ancestor 
of $\pi$ is labelled with an ID using a state different from $q_0'$; and (2) there is 
some $1 \leq j \leq |w|$ such that $\pi_f, \pi_f',\pi_f'',\dots $ is the maximal path of nodes (up to a leaf) labelled with 
$(q_m,w,j)$, $(q_m,w,j+1),\dots,(q_m,w,|w|+1)$, and where the father of $(q_m,w,j)$ is not a configuration using state $q_m$. 
Property (1) represents the automaton searching for its starting point, and (2) represents the 
end of the computation of the part of $\A_\exp'$ that is representing the non-nesting part of $\exp$. 
We denote such nodes $\pi_s$ and $\pi_f$ as the {\em tacit start} and {\em tacit ending} of $\Pi$.

With this definitions we can show the following.

\begin{lemma}
\label{lem-1-correctbranch}
Let $\exp$ a NRE, $A_\exp$ the automaton constructed for $\exp$, 
$T$ a graph over $\Sigma$ that is a $k$-branch semipath, where $k$ is the nesting depth of $\exp$, 
and $w = \trans(T)$ be the encoding of $T$ as a string. 
Then a pair $(u,v)$ belongs to $\sem{\exp}_T$ if and only if 
there is an accepting computation tree of $A_\exp$ on input $w$ whose tacit start is 
labelled with $(q_0,w,\posi(u))$ and whose tacit ending is labelled with  $(q_m,w,\posi(v))$. 
\end{lemma}

\begin{proof}
Let $T$ be a $k$-branch semipath and let $A_\exp = (Q,U,q_0,F,\Gamma,\delta)$ constructed as above. 
Let us start with the {\bf Only if direction}. Assume that $\sem{\exp}_T$ contains the pair  $(u,v)$ for some nodes $u$ and 
$v$ of $T$. We show the statement of the Lemma by induction on $R$. 

We only show the case when $\exp = a$. The case when $\exp = a^-$ is completely symmetrical, and the remaining ones 
follow from the proof of Lemma \ref{lem-1-correctstring}.  
\begin{itemize}
\item If $R = a$ for some $a \in \Sigma'$ and $(u,v)$ belong to $\sem{\exp}_T$, then either $u$ is a prefix of $v$ and the edge 
$(u,a,v)$ is in $T$, or $v$ is a prefix of $u$ and the edge $(v,a^-,u)$ is in $T$. 
For the former case, assume that $u = w \cdot i$, all children of $u$ are $u_1,\dots,u_n$, and $v = u_\ell$ for some $1 \leq \ell \leq n$. If 
$v = w \cdot i \cdot i$, then starting in $\posi(u)$ one can use the transitions that loop in some of the $B_i$'s until we reach symbol 
$(i,a,i)$ in $\trans(T)$, from which we advance to the final state of $A_\exp$. Otherwise, If $v = w \cdot i \cdot j$ for some $i < j$, 
we can also loop, but this time until we advance to the final state by means of symbol $(i,a,j)$. 
For the latter case, assume that $v = w \cdot i$,  all children of $v$ are $v_1,\dots,v_n$, and $u = v_\ell$ for some $1 \leq \ell \leq n$. If 
$u = w \cdot i \cdot i$, then by definition the symbol $(i,a^-,i)$ is directly before $\posi(u)$. We can then non-deterministically jump 
to $q_r^1$ in $A_\exp$, check that effectively the symbol $(i,a^-,i)$ exists, and move backwards according to the transitions looping in 
the copies of automata $B_j^-$'s, until we reach $\posi(v)$. Otherwise if $u = w \cdot i \cdot j$ with $i < j$ then by definition again 
the symbol $(i,a^-,j)$ is directly before $\posi(u)$, and we continue along the same lines as before. 
\end{itemize}

Next, for the {\bf If} direction, assume there is an accepting computation tree of $A_\exp'$ on input $w$ whose tacit start is 
labelled with $(q_0,w,\posi(u))$ and whose tacit ending is labelled with  $(q_m,w,\posi(v))$. We show that $(u,v) \in \exp_R$ by induction. 
Once again, it suffices to show the base case when $\exp = a$. 

\begin{itemize}
\item If $R = a$ for some $a \in \Sigma'$, there are two types of computation tree for $A_\exp$. 
Assume first that such tree does not mention any node labelled with an 
ID that corresponds to $q_r^1$ or $q_r^2$. Then, at some point in the computation tree, there must be a jump from an 
ID of form $(q_0,w,i)$ to an ID of form $(q_f,w,j)$ for some positions $i$ and $j$ in $w$, and when reading a symbol of form 
$(\ell_1,a,\ell_2)$. From the construction of $A_R$, 
we can only loop from state $q_f$ if we are directly after a symbol of form $(\ell', \finbra,\ell')$, and thus in this case we can not loop; it must be that 
$j = \posi(v)$. Furthermore, if one stays in state $q_0$ one can only move forward, in a way that the subword 
between position $\posi(u)$ and $i$ must correspond to 
a concatenation of word in some $L_{(i,j)}$s. Then either $i = \posi(u)$ or $u$ has at least two children, and position  
$i$ corresponds to the position in $w$ after we have read the encoding for some of these children. It then follows from our translation $\trans$ 
that the edge between $u$ and $v$ in $T$ is labelled $a$.
 
Next, assume that the tree does mention an ID going through $q_r^1$. In this case, there must be a step from $q_0$ to $q_r^1$ that 
is a move backwards, and 
then to advance to $q_r^2$ we need a symbol of form $(\ell_1,a^-,\ell_2)$. In other words, at some point in the tree we move from 
ID $(q_0,w,i)$ to $(q_r^1,w,i-1)$ and then to $(q_r^2,w,i-1) $, and such that the symbol between positions $i-1$ and $i$ is 
of the form $(\ell_1,a^-,\ell_2)$. It follows that $i = \posi(u)$, since by moving forwards in $q_0$ we shall never reach a point directly after 
a symbol with this form. A similar argument as the previous case also shows that either $\posi(v) = i-1$ ot the 
path from $\posi(v)$ to $i-1$ must correspond to 
a concatenation of words in $L_{(i,j)}$'s. By inspecting our translation, we then have that there must be an edge $(v,a^-,u)$ in $T$, and 
therefore $(u,v) \in \sem{\exp}_T$. 
\end{itemize}
\end{proof}

Just as we saw for the case of semipaths, we need to be more careful, and explicitly mark with symbols $S$ and $E$ to positions 
in $\trans(T)$, in order to distinguish the root and leaves of $T$ with the actual $k$-branch semipath that is framed by nodes 
$u$ and $v$. Formally, given a $k$-branch semipath $T$, and two nodes $u$ and $v$ of $T$, the expansion 
$T[u \rightarrow S,v \rightarrow E]$ is the $k$-branch semipath defined as follows. 
If $u = w \cdot i$, and its children $w \cdot i \cdot \ell_1, \dots, w \cdot i \cdot \ell_n$, then rename all children 
to $w \cdot i \cdot i \cdot \ell_1, \dots, w \cdot i \cdot i \cdot \ell_n$, and all of the descendants of $u$ accordingly, so that the 
domain remains prefix-closed. Now $u$ has a single child, $w \cdot i \cdot i$ connected by an edge labelled $S$, and this node 
is the father of all the nodes that were previously childrens of $u$. Repeat with $v$ and $E$. 
The intuition is that $T[u \rightarrow S,v \rightarrow E]$ is created by replacing node $u$ in $T$ with an edge labelled by $S$, and 
node $v$ by an edge labelled $E$. 
Let $\exp$ be an NRE. Using the ideas presented in the proof of Lemma \ref{lem-2-correctstring} it is not difficult to 
define a translation from $A_\exp$ to an automaton $A_\exp^{S,E}$ such that the following holds: 

\begin{lemma}
\label{lem-2-correctbranch}
Let $T$ be a $k$-branch semipath, and $R$ an NRE. Then a pair $(u,v)$ belongs to $\sem{\exp}_T$ if and only if 
$A_\exp^{S,E}$ accepts the semipath $T[u \rightarrow S,v \rightarrow E]$. 
\end{lemma}

We can now state our algorithm for solving {\sc SP-QueryContainment}. 
On input NREs $\exp_1$ and $\exp_2$ over $\Sigma$, we perform the following operations: 
\begin{enumerate}
\item Compute an NFA $\A^{S,E}$ that accepts only those words over 
$(\Sigma \cup \{S,E\})^*$ of form $w_1Sw_2Ew_3$, for each $w_1,w_2,w_3$ in $\Sigma^*$.  
\item Compute $\A_k$ that accepts only those words which are translations of 
$k$-branch semipaths over $\Sigma$, where $k$ is the nesting depth of $\exp_1$. 
\item Compute $\A_{\exp_1}^{S,E}$ and $A_{\exp_2}^{S,E}$ as explained above. 
\item Compute the A2FA $A^c = (A_{\exp_2}^{S,E})^c$ whose language is the complement of $A_{\exp_2}^{S,E}$ 
\item Compute the A2FA $A$ whose language is the intersection of the languages $A^{S,E}$, $A_{\exp_1}^{S,E}$,  $A^c$ 
and $A_\trans$. 
\item Check that the language of $A$ is empty
\end{enumerate}

We have seen how to perform the second step in polynomial time, and steps (1), (3), (4) can be easily 
performed in \ptime\ using standard techniques from automata theory. Finally, Proposition \ref{prop-nonempt-a2fa} shows that step (5) can be performed in 
\pspace. Thus, all that is left to prove is 
that the language of the resulting automata $A$ is empty if and only if $\exp_1 \subseteq \exp_2$. 

Assume first that $\exp_1 \subseteq \exp_2$, and 
assume for the sake of contradiction that there is a word $w \in L(A)$. This word is then accepted by $\A^{S,E}$ and $\A_k$. 
Then there is a $k$-branch semipath 
$T$ over $\Sigma$ and two nodes $u$ and $v$ of $T$ such that $w = \trans(T[u \rightarrow S,v \rightarrow E])$. Furthermore,  
$w$ is accepted by $A_{\exp_1}^{S,E}$, but not by $A_{\exp_2}^{S,E}$. This implies, by Lemma \ref{lem-2-correctbranch}, 
that $(u,v) \in \sem{\exp_1}_T$ but $(u,v) \in \sem{\exp_2}_T$, which is a contradiction. 

On the other hand if $L(A)$ is empty but $\exp_1 \not\subseteq \exp_2$, then for some 
$k$-branch semipath $T$ that is canonical for $\exp_1$ and two nodes $u$ and $v$ of $T$ we have that 
$(u,v) \in \sem{\exp_1}_T$ yet $(u,v) \notin \sem{\exp_2}_T$. By Lemma \ref{lem-2-correctbranch} 
we have that $w = \trans(T[u \rightarrow S,v \rightarrow E])$ is such that 
$w$ belongs to $\A^{S,E}$, $\A_k$ and $A_{\exp_1}^{S,E}$, and it is not accepted by $A_{\exp_2}^{S,E}$, thus belonging to $A^c$. 
This means that $w$ is in the language of $A$, which is a contradiction.

\bibliographystyle{abbrv}
\bibliography{biblio}

\end{document}